\documentclass{article}

\usepackage{macros}
\usepackage[in]{fullpage}

\usepackage{tikz}

\newcommand{\fsp}{\cc{FS_2P}}
\newcommand{\fsigp}{\cc{F\Sigma_2P}}

\title{Symmetric Exponential Time Requires Near-Maximum Circuit Size: Simplified, Truly Uniform}
\author{Zeyong Li\thanks{National University of Singapore. Email: li.zeyong@u.nus.edu}}
\date{}

\begin{document}

\maketitle

\begin{abstract}
    In a recent breakthrough, Chen, Hirahara and Ren \cite{CHR23} prove that $\ensuremath{\mathsf{S_2E}}/_1 \not\subset \ensuremath{\mathsf{SIZE}}[2^n/n]$ by giving a single-valued $\ensuremath{\mathsf{FS_2P}}$ algorithm for the Range Avoidance Problem ($\ensuremath{\mathsf{Avoid}}$) that works for infinitely many input size $n$.

    Building on their work, we present a simple single-valued $\ensuremath{\mathsf{FS_2P}}$ algorithm for $\ensuremath{\mathsf{Avoid}}$ that works for all input size $n$. As a result, we obtain the circuit lower bound $\ensuremath{\mathsf{S_2E}} \not\subset \ensuremath{i.o.}$-$\ensuremath{\mathsf{SIZE}}[2^n/n]$ and many other corollaries:
    \begin{enumerate}
        \item Almost-everywhere near-maximum circuit lower bound for $\ensuremath{{\mathsf{\Sigma_2E}}} \cap \ensuremath{{\mathsf{\Pi_2E}}}$ and $\ensuremath{\mathsf{ZPE}}^{\mathsf{NP}}$.
        \item Pseudodeterministic $\ensuremath{\mathsf{FZPP}}^{\mathsf{NP}}$ constructions for combinatorial objects such as: Ramsey graphs, rigid matrices, pseudorandom generators, two-source extractors, linear codes, hard truth tables, and $\ensuremath{K^{\mathsf{poly}}}$-random strings.
    \end{enumerate}
\end{abstract}

\section{Introduction}
Proving circuit lower bounds has been one of the most fundamental problems in complexity theory, and has close connections to many other fundamental questions such as $\cc{P}$ versus $\NP$, derandomization and so on. For instance, if we could show $\cc{E} \not\subset \cc{SIZE}[2^{o(n)}]$, then we achieve unconditional derandomization, i.e. $\cc{prBPP} = \cc{prP}$ \cite{NW94,IW97}. Morally speaking, it is a quest to distinguish the computational power between uniform and non-uniform computations. 

In the search of exponential circuit lower bound, we know that almost all $n$-bit boolean functions requires near-maximum $(2^n/n)$-sized circuit via a simple counting argument \cite{Shannon, FM05}. While such argument is inherently non-constructive, it serves as some form of evidence that we should be optimistic about finding one such boolean function in some not-too-large complexity class. 

However, limited progress has been made over the past few decades in the search of a small complexity class with exponential circuit lower bound. In 1982, Kannan \cite{Kannan82} showed that $\cc{\Sigma_3E} \cap \cc{\Pi_3E}$ contains a language with maximum circuit complexity. The frontier was later pushed to $\cc{\Delta_3E} = \cc{E}^{\cc{\Sigma_2P}}$ by Miltersen, Vinodchandran and Watanabe \cite{MVW99}, which persisted to be the state of the art for more than twenty years.

Very recently in a breakthrough result, Chen, Hirahara and Ren \cite{CHR23} prove that: 
\[ \cc{S_2E}/_1 \not\subset \cc{SIZE}[2^n/n] \; . \]

That is, the symmetric time class $\cc{S_2E}$ with one bit of advice requires near-maximum circuit complexity. 
Despite requiring one bit of advice, this is a huge improvement compared to the previous result since $\cc{S_2E} \subseteq \cc{ZPE}^{\NP} \subseteq \cc{\Sigma_2E}$ is expected to be a much smaller class compared to $\cc{\Delta_3E} = \cc{E}^{\cc{\Sigma_2P}}$, assuming that the exponential hierarchy does not collapse. And it turns out that this exciting result is indeed a corollary of a \emph{single-valued} algorithm for the Range Avoidance problem ($\cc{Avoid}$). 

\paragraph{The Range Avoidance Problem}
In the past few years, study on the range avoidance problem has been another exciting line of research \cite{KKMP21,Korten,GLW22,RSW22,CHLR23,GGNS23,ILW23}. The problem itself is defined as follows: given an expanding circuit $C:\bits^n \rightarrow \bits^{n+1}$, find a string not in the image of $C$. That is, output $y \in \bits^{n+1}$ where $\forall x \in \bits^n, C(x) \neq y$.

At a first glance, $\cc{Avoid}$ might seem to be a easy problem: a random string $y$ would be a non-image of $C$ with probability at least $1/2$. However, it is unclear how to amplify the probability of success without an $\NP$ oracle. And given an $\NP$ oracle, $\cc{Avoid}$ can be trivially solved in $\cc{FZPP}^\NP$. Somewhat surprisingly, Korten \cite{Korten} showed that $\cc{Avoid}$ is as hard as finding optimal explicit constructions of important combinatorial objects such as Ramsey graphs \cite{Radziszowski2021}, rigid matrices \cite{GLW22,GGNS23}, pseudorandom generators \cite{CT21}, two-source extractors \cite{CZ19,Li23}, linear codes \cite{GLW22}, hard truth tables \cite{Korten}, and strings with maximum time-bounded Kolmogorov complexity ($K^{\poly}$-random strings) \cite{RSW22}. Therefore, finding any non-trivial algorithm for $\cc{Avoid}$ implies algorithms for constructing these important objects.

\paragraph{Single-Valued Algorithm} Let $\Pi$ be a search problem where $\Pi_x$ denotes the set of solutions for input $x$. Morally speaking, A single-valued algorithm $A$ on input $x$ succeeds only when it outputs some canonical solution $y_x \in \Pi_x$. Here are two examples:
\begin{itemize}
    \item A single-valued $\cc{FNP}$ algorithm should have at least one successful computational path and should output $\bot$ in all other computational paths. (Studied as $\cc{NPSV}$ constructions in, e.g. \cite{HNOS96}).
    \item A single-valued $\cc{FZPP}$ algorithm should succeed on most (e.g. $\geq 2/3$ fraction) computational paths and output $\bot$ otherwise. (Studied as pseudodeterministic constructions in, e.g. \cite{GG11}).
\end{itemize}
In particular, the trivial $\cc{FZPP}^{\NP}$ algorithm for $\cc{Avoid}$ (i.e. sample a string and check with the $\NP$ oracle) is inherently not single-valued: the outputs would be different in almost all executions!

As pointed out in \cite{CHR23} and many other previous works, circuit lower bounds can be viewed as single-valued construction of hard truth tables. In particular, if for all input size $n$, one could compute (consistently the same) truth table that is hard against all $s(n)$-size circuits, then the language whose characteristic function is set to the truth table, would be a hard language $\notin \cc{SIZE}[s(n)]$. Given that finding hard truth table reduces to $\cc{Avoid}$, this connects the two tasks: proving circuit lower bound and finding single-valued algorithm for $\cc{Avoid}$.

\subsection{Our Results}

\subsubsection{Algorithm for the Range Avoidance Problem}

\cite{CHR23} presented a single-valued $\fsp$ algorithm for $\cc{Avoid}$ that works \emph{infinitely often}. I.e., for infinitely many (but unknown) input size $n$, there is a $\cc{S_2P}$ machine\footnote{To provide some intuition for readers who are unfamiliar with this class, $\cc{S_2P}$ is contained in $\cc{ZPP}^{\NP}$\cite{Cai07}.}, on input a circuit $C:\bits^n \rightarrow \bits^{n+1}$, outputs a canonical non-image of $C$ in the successful computational paths and $\bot$ otherwise.

In this work, we extend their algorithm to a single-valued $\fsp$ algorithm for $\cc{Avoid}$ that works \emph{for all input size $n$}.


\begin{theorem}
    There is a single-valued $\fsp$ algorithm $A$: when given any circuit $C:\bits^n \rightarrow \bits^{n+1}$ as input, $A(C)$ outputs $y_C$ such that $y_C \notin \Im(C)$ .
\end{theorem}

The remaining results are corollaries of our new single-valued $\fsp$ algorithm for $\cc{Avoid}$.

\subsubsection{Almost-everywhere near-maximum circuit lower bounds}
In \cite{CHR23}, their circuit lower bound requires one bit of advice. This is because their $\fsp$ algorithm for $\cc{Avoid}$ only works for infinitely many input size $n$, and the one bit of advice is necessary for indicating which input size is `good'. As our algorithm works for all input size $n$, the circuit lower bound that we obtain are almost-everywhere and completely removes any non-uniform advice:

\begin{theorem}
    $\cc{S_2E} \not\subset i.o.$-$\cc{SIZE}[2^n/n]$. Moreover, this holds in every relativized world.
\end{theorem}

Similar to \cite{CHR23}, our results fully relativize.

Via known results where $\cc{S_2E} \subseteq \cc{ZPE}^{\NP}$ \cite{Cai07} and $\cc{ZPE}^{\NP} \subseteq \cc{\Sigma_2E} \cap \cc{\Pi_2E}$, we obtain the following corollaries:

\begin{corollary}
    $\cc{ZPE}^{\NP} \not\subset i.o.$-$\cc{SIZE}[2^n/n]$. Moreover, this holds in every relativized world.
\end{corollary}

\begin{corollary}
    $\cc{\Sigma_2E} \cap \cc{\Pi_2E} \not\subset i.o.$-$\cc{SIZE}[2^n/n]$. Moreover, this holds in every relativized world.
\end{corollary}

\subsubsection{Explicit constructions}
Next, by Cai's theorem \cite{Cai07}, we have $\cc{S_2P} \subseteq \cc{ZPP}^\NP$. Hence, $\cc{Avoid}$ and all explicit construction problems admit a single-valued $\cc{FZPP}^{\NP}$ algorithm. Equivalently speaking, $\cc{Avoid}$ and all explicit construction problems admit a pseudodeterministic (with an $\NP$ oracle) algorithm, where a pseudodeterministic algorithm for a search problem is a probabilistic algorithm that with high probability outputs a fixed solution on any given input.

\begin{theorem}
    There is a single-valued $\cc{FZPP}^{\NP}$ algorithm $A$: when given any circuit $C:\bits^n \rightarrow \bits^{n+1}$ as input, $A(C)$ outputs $y_C$ with probability at least $2/3$ and $y_C \notin \Im(C)$.
\end{theorem}

\begin{corollary}[Informal]
    There are zero-error pseudodeterministic constructions for the following objects with an $\NP$ oracle for every input size $n$: Ramsey graphs, rigid matrices, pseudorandom generators, two-source extractors, linear codes, hard truth tables, and $K^{\poly}$-random strings.
\end{corollary}

\subsubsection{Missing-String problem}
Lastly, we point out a connection to the $\cc{MissingString}$ problem. The $\cc{MissingString}$ problem is defined as follows: given a list of $m$ strings $x_1, \ldots, x_m \in \bits^n$ where $m < 2^n$, the goal is to output any string $y \in \bits^n$ not in the list. 

In \cite{VW23}, Vyas and Williams connected the circuit complexity of the $\cc{MissingString}$ with the (relativized) circuit complexity of $\cc{\Sigma_2E}$.

\begin{theorem}[\cite{VW23}]
    The following are equivalent:
    \begin{enumerate}
        \item $\cc{\Sigma_2E}^A \not\subset$ i.o.-$\cc{SIZE}^A[2^{\Omega(n)}]$ for every oracle A;
        \item for $m = 2^{\Omega(n)}$, the $\cc{MissingString}$ problem can be solved by a uniform family of size-$2^{\poly(n)}$ depth-3 $\cc{AC}^0$ circuits.
    \end{enumerate}
\end{theorem}

As a corollary of our circuit lower bound which relativizes, we can conclusively claim that:

\begin{corollary}
    For $m = 2^{\Omega(n)}$, the $\cc{MissingString}$ problem can be solved by a uniform family of size-$2^{\poly(n)}$ depth-3 $\cc{AC}^0$ circuits.
\end{corollary}

\subsection{Proof Overview}

\subsubsection{Korten's Reduction}

We start by reviewing Korten's reduction from \cite{Korten}, which reduces $\cc{Avoid}$ on a circuit with $n$-bit stretch (i.e. maps $n$-bit input to $2n$-bit output) to $\cc{Avoid}$ on some other circuit with much longer stretch.


Given any circuit $C:\bits^{n} \rightarrow \bits^{2n}$ and parameter $T = n\cdot 2^k$, Korten builds another circuit $\cc{GGM}_T[C]:\bits^n \rightarrow \bits^T$ by applying the circuit $C$ in a perfect binary tree:
\begin{enumerate}
    \item Assign the root vertex $(0,0)$ with value $v_{0,0} = x$. Build a perfect binary tree of height $k$ where $(i,j)$ denotes the $j$th vertex on the $i$th level for $0 \le i \le k$ and $0 \le j \le 2^i - 1$. 
    \item For each vertex $(i,j)$ on the tree, evaluate $y = C(v_{i,j})$ and assign its left child with the first $n$ bits of $y$ and its right child with the last $n$ bits of $y$. 
    \item The output of $\cc{GGM}_T[C](x)$ is simply the concatenation of the values of the $2^k$ leaves.
\end{enumerate}

\begin{figure}[h]
	\begin{center}
			\begin{tikzpicture}[
					level distance=1cm,
					level 1/.style={sibling distance=4cm},
					level 2/.style={sibling distance=2cm},
					level 3/.style={sibling distance=1cm},
					every node/.style={rectangle,draw}
					]
					
\node {$v_{0,0}$}
    child { node {$v_{1,0}$}
        child { node {$v_{2,0}$}
            child { node {$v_{3,0}$}}    
            child { node {$v_{3,1}$}}
        }    
        child { node {$v_{2,1}$}
            child { node {$v_{3,2}$}}    
            child { node {$v_{3,3}$}}
        }
    }    
    child { node {$v_{1,1}$}
        child { node {$v_{2,2}$}
            child { node {$v_{3,4}$}}    
            child { node {$v_{3,5}$}}
        }    
        child { node {$v_{2,3}$}
            child { node {$v_{3,6}$}}    
            child { node {$v_{3,7}$}}
        }
    };    
					
				\end{tikzpicture}
		\end{center}
	\caption{An illustration of a GGM tree of height $3$.}
\end{figure}
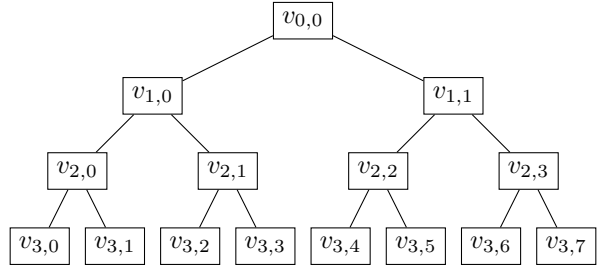



Notice that on any fixed input $x\in \bits^n$, every vertex on the GGM tree has an $n$-bit value. Hence, we call it a \emph{fully-assigned} GGM tree. It is not hard to see that one can efficiently (takes time linear in the height of the tree, $k$) evaluate the assigned value at any vertex by traversing the tree and apply the circuit $C$ at most $k$ times. In other words, a \emph{fully-assigned} GGM tree has small circuit complexity.

Korten's reduction asserts that given an $\NP$ oracle and any $f\in \bits^T \backslash \Im(\cc{GGM}_T[C])$, there is a deterministic algorithm that finds a non-image of $C$ and runs in time $\poly(T,n)$. And the algorithm is simple: 
\begin{enumerate}
    \item Set the assigned values of the leaves to be $f$.
    \item Next, traverse the tree in a simple bottom up manner. I.e. traverse the $2^{k-1}$ vertices on the $(k-1)$th level one by one (say, from right to left), then proceed to the $(k-2)$th level and so on, until reaching root.
    \item For each interval vertex $u$ traversed, assign $v_u$ with the lexicographically first\footnote{the original reduction from \cite{Korten} didn't specify the lexicographically first requirement. This requirement is important for \cite{CHR23} and us in order to obtain a single-valued algorithm, and it comes for free given the $\NP$ oracle.} $n$-bit string $x$ such that $C(x)$ correctly evaluates to the assigned values of $u$'s children (Note that this step uses the $\NP$ oracle).
    \item Whenever such string cannot be found, we successfully find a non-image of $C$ (i.e. the assigned values of $u$'s children). The algorithm now returns with the non-image of $C$ found and assigns $\bot$ to all remaining vertices.
\end{enumerate}

\begin{figure}[h]
	\begin{center}
			\begin{tikzpicture}[
					level distance=1cm,
					level 1/.style={sibling distance=4cm},
					level 2/.style={sibling distance=2cm},
					level 3/.style={sibling distance=1cm},
					every node/.style={rectangle,draw}
					]
					
  \node {$\bot$}
    child { node {$\bot$}
        child { node {$\bot$}
            child { node {$v_{3,0}$}}    
            child { node {$v_{3,1}$}}
        }    
        child { node[very thick] {$\bot$}
            child { node {$v_{3,2}$}}    
            child { node {$v_{3,3}$}}
        }
    }    
    child { node {$\bot$}
        child { node {$v_{2,2}$}
            child { node {$v_{3,4}$}}    
            child { node {$v_{3,5}$}}
        }    
        child { node {$v_{2,3}$}
            child { node {$v_{3,6}$}}    
            child { node {$v_{3,7}$}}
        }
    };    
					
				\end{tikzpicture}
		\end{center}
\caption{An illustration of the \emph{partially-assigned} GGM tree from running Korten's algorithm.}
\end{figure}
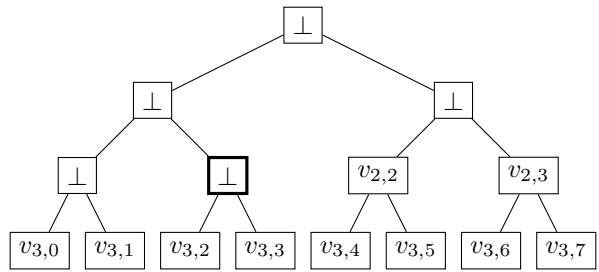


\subsubsection{History of the Reduction}

The computational history of Korten's reduction on fixed input $(C,f)$ is fully characterised by a \emph{partially-assigned} GGM tree i.e. some of the vertices are assigned $\bot$. We are interested in the computaional history because it has a few nice properties:

\begin{enumerate}
    \item \textbf{Contains a Canonical Solution:} Notice that the execution of Korten's algorithm is fully deterministic. Hence, it produces the same non-image of $C$ given the same $f$. And the solution is clearly stored in the \emph{partially-assigned} GGM tree. 
    \item \textbf{Locally Verifiable:} Every step of execution is very simple, making them locally verifiable. In other words, to verify any particular step of the execution, we only need to look at a constant number of assigned values on the \emph{partially-assigned} GGM tree.
\end{enumerate}

Moreover, by choosing $T = 2n\cdot 2^{2n}$, we know an $f$ that is trivially not in the image of $\cc{GGM}_T[C]$: the concatenation of all $2n$-length strings. 

\subsubsection{Finding a Short Description of the History}

The downside of choosing $T = 2n\cdot 2^{2n}$ is that the size of the computational history is now exponential in $n$. The locally verifiable property allows us to use a universal quantifier ($\forall$) and a $O(\log T)$-bit variable to verify all $O(T)$ steps of the algorithm, but ultimately we need a short ($\poly(n)$) description of the computational history if we want to, for example build a $\fsigp$ algorithm. 

The authors in \cite{CHR23} appeal to the iterative win-win argument for such a short description: they manage to show that within a large interval of input size, there exists at least one input size $n$ such that the corresponding computational history admits a short description. In fact, the computational history will be the output of a (different) \emph{fully-assigned} GGM tree, leveraging the fact that \emph{fully-assigned} GGM tree has small circuit complexity.


We take a slightly different approach: the key observation is that, by changing traversal order in Korten's algorithm, the resulting computational history (i.e. a \emph{partially-assigned} GGM tree) also has small circuit complexity! More specifically, if we change the traversal order to a post-order traversal (i.e. traverse the left subtree, then the right subtree, and finally the root), the resulting \emph{partially-assigned} GGM tree can be decomposed into $O(n)$ smaller \emph{fully-assigned} GGM trees. See \cref{fig:modified} for an illustration. The roots of the \emph{fully-assigned} GGM trees are drawn in circles.

As such, we obtain a short description of the computational history: simply store the roots of all these $O(n)$ \emph{fully-assigned} GGM tree.

\begin{figure}[h]
	\begin{center}
			\begin{tikzpicture}[
					level distance=1cm,
					level 1/.style={sibling distance=4cm},
					level 2/.style={sibling distance=2cm},
					level 3/.style={sibling distance=1cm},
					every node/.style={rectangle,draw}
					]
					
  \node {$\bot$}
    child { node[circle] {$v_{1,0}$}
        child { node {$v_{2,0}$}
            child { node {$v_{3,0}$}}    
            child { node {$v_{3,1}$}}
        }    
        child { node {$v_{2,1}$}
            child { node {$v_{3,2}$}}    
            child { node {$v_{3,3}$}}
        }
    }    
    child { node {$\bot$}
        child { node[circle] {$v_{2,2}$}
            child { node {$v_{3,4}$}}    
            child { node {$v_{3,5}$}}
        }    
        child { node[very thick] {$\bot$}
            child { node[circle] {$v_{3,6}$}}    
            child { node[circle] {$v_{3,7}$}}
        }
    };    
					
				\end{tikzpicture}
		\end{center}
\caption{An illustration of the \emph{partially-assigned} GGM tree from running modified Korten's algorithm.}
\label{fig:modified}
\end{figure}
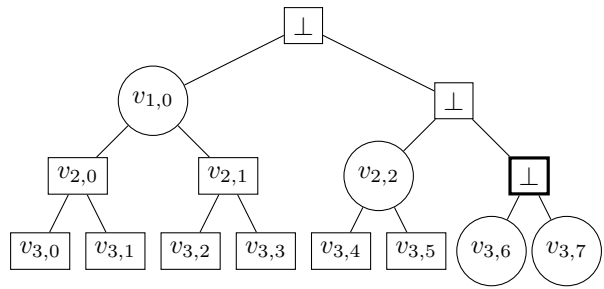


\subsubsection{Generalizing to $\fsp$}
All the ingredients above allow us to build a single-valued $\fsigp$ algorithm for $\cc{Avoid}$. In order to generalise it to a $\fsp$ algorithm, we need a selector algorithm that picks the better witness (in this case, the correct description of the computational history). Now that we have a small description, this turns out to be an easy task. It is now easy to identify a single vertex with different assigned values in the two histories, and traverse down the tree until we hit the leaves, where we know the correct assigned value (i.e. the concatenation of all $2n$-length strings).

\section{Preliminaries}

\begin{definition}
    Let $s: \N \rightarrow \N$. We say that a language $L \in \cc{SIZE}[s(n)]$ if $L$ can be computed by circuit families of size $O(s(n))$ for all sufficiently large input size $n$.
\end{definition}

\begin{definition}
    Let $s: \N \rightarrow \N$. We say that a language $L \in i.o.\text{-}\cc{SIZE}[s(n)]$ if $L$ can be computed by circuit families of size $O(s(n))$ for infinitely many input size $n$.
\end{definition}

By definition, we have $\cc{SIZE}[s(n)] \subseteq i.o.\text{-}\cc{SIZE}[s(n)]$. Hence, circuit lower bounds against $i.o.\text{-}\cc{SIZE}[s(n)]$ are stronger and sometimes denoted as \emph{almost-everywhere} circuit lower bound in the literature. 

\begin{definition}
    The Range Avoidance ($\cc{Avoid}$) problem is defined as follows: given as input the description of a Boolean circuit $C:\bit^n \rightarrow \bit^{m}$, for $m > n$, find a $y \in \{0, 1\}^m$ such that $\forall x \in \{0, 1\}^n: C(x) \neq y$.
\end{definition}

We assume basic familiarity with computational complexity theory, such as complexity classes in the polynomial hierarchy (see e.g. \cite{AB09,Gol08} for references). 

\begin{definition}
    Let $T:\N \rightarrow \N$. We say that a language $L\in \cc{S_2TIME}[T(n)]$, if there exists an $O(T(n))$-time verifier $V(x, \pi_1, \pi_2)$ that takes $x\in\bits^n$ and $\pi_1, \pi_2 \in \bits^{T(n)}$ as input, satisfying that:
    \begin{itemize}
        \item if $x\in L$, then there exists $\pi_1$ such that for every $\pi_2$, \\ $V(x, \pi_1, \pi_2) = 1$, and
        \item if $x\notin L$, then there exists $\pi_2$ such that for every $\pi_1$,\\ $V(x, \pi_1, \pi_2) = 0$.
    \end{itemize}
    Moreover, we say $L \in \cc{S_2E}$ if $L\in \cc{S_2TIME}[T(n)]$ for some $T(n)\leq 2^{O(n)}$, and  $L \in \cc{S_2P}$ if $L\in \cc{S_2TIME}[p(n)]$ for some polynomial $p$.
\end{definition}

A search problem $\Pi$ maps every input $x \in \bits^*$ into a solution set $\Pi_x \subseteq \bits^*$. An algorithm $A$ solves the search problem $\Pi$ on input $x$ if $A(x) \in \Pi_x$.

\begin{definition}[Single-valued $\fsigp$ algorithm]
    A single-valued $\fsigp$ algorithm $A$ is specified by a polynomial $\ell(\cdot)$ together with a polynomial-time algorithm $V_A(x, \pi_1, \pi_2)$. On an input $x\in\bits^*$, we say that $A$ outputs $y_x \in \bits^*$, if the following hold:
    \begin{enumerate}
        \item There is a $\pi_1 \in \bits^{\ell(|x|)}$ such that for every $\pi_2 \in \bits^{\ell(|x|)}$, $V_A(x,\pi_1, \pi_2)$ outputs $y_x$.
        \item For every $\pi_1 \in \bits^{\ell(|x|)}$ there is a $\pi_2 \in \bits^{\ell(|x|)}$, such that $V_A(x,\pi_1, \pi_2)$ outputs either $y_x$ or $\bot$.
    \end{enumerate}
    And we say that $A$ solves a search problem $\Pi$ if on any input $x$ it outputs a string $y_x$ and $y_x \in \Pi_x$.
    
\end{definition}
\begin{definition}[Single-valued $\fsp$ algorithm]
    A single-valued $\fsp$ algorithm $A$ is specified by a polynomial $\ell(\cdot)$ together with a polynomial-time algorithm $V_A(x, \pi_1, \pi_2)$. On an input $x\in\bits^*$, we say that $A$ outputs $y_x \in \bits^*$, if the following hold:
    \begin{enumerate}
        \item There is a $\pi_1 \in \bits^{\ell(|x|)}$ such that for every $\pi_2 \in \bits^{\ell(|x|)}$, $V_A(x,\pi_1, \pi_2)$ outputs $y_x$.
        \item There is a $\pi_2 \in \bits^{\ell(|x|)}$ such that for every $\pi_1 \in \bits^{\ell(|x|)}$, $V_A(x,\pi_1, \pi_2)$ outputs $y_x$.
    \end{enumerate}
    And we say that $A$ solves a search problem $\Pi$ if on any input $x$ it outputs a string $y_x$ and $y_x \in \Pi_x$.
\end{definition}

\section{Modified Korten's reduction}

\paragraph{Notation.} We follow the notations from \cite{CHR23} closely: Let $s$ be a $n$-bit string. We use 0-index where $s_0$ denotes the first bit of $s$ and $s_{n-1}$ denotes the last bit of $s$. Let $i < j$, we use $s_{[i,j)}$ to denote the substring of $s$ from the $i$th bit to the $(j-1)$th bit. We use $s_1 \circ s_2$ to denote the concatenation of two strings $s_1$ and $s_2$.

We identify any vertex in a perfect binary tree of height $2n+1$ with a tuple $(i,j)$ where $i \in [0, 2n+1]$ and $j \in [0, 2^i - 1]$, indicating that the vertex is the $j$th vertex on level $i$. Note that the two children of $(i,j)$ are $(i+1, 2j)$ and $(i+1, 2j+1)$.

\subsection{The GGM Tree}
Recall that the GGM tree construction from \cite{GGM86}  (vaguely speaking) increases the stretch of a circuit $C:\bits^n \rightarrow \bits^{2n}$ to arbitrarily long by applying $C$ in a perfect binary tree manner. 

\begin{definition}[The GGM tree construction \cite{GGM86}]
    Let $C:\bits^n \rightarrow \bits^{2n}$ be a circuit. Let $n, T \in \N$ be such that $T \geq 4n$ and let $k$ be the smallest integer such that $2^kn \geq T$. The function $\cc{GGM}_T[C]:\bits^n \rightarrow \bits^T$ is defined as follows.
    
    Consider a perfect binary tree with $2^k$ leaves, where the root is on level $0$ and the leaves are on level $k$. Each node is assigned a binary string of length $n$, and for $0\leq j < 2^i$, denote $v_{i,j} \in \bits^n$ the value assigned to the vertex $(i,j)$ (i.e. $j$-th node on level $i$). Let $x\in \bits^n$. We perform the following computation to obtain $\cc{GGM}_T[C](x)$: we set $v_{0,0}:=x$, and for each $0 \leq i <k, 0 \leq j < 2^i$, we set $v_{i+1, 2j}:=C(v_{i,j})_{[0,n)}$ (i.e. the first half of $C(v_{i,j})$) and $v_{i+1, 2j+1}:=C(v_{i,j})_{[n,2n)}$ (i.e. the second half of $C(v_{i,j})$). 

    Finally, we concatenate all values of the leaves and take the first $T$ bits as the output:
    \[ \cc{GGM}_T[C](x) := (v_{k,0} \circ v_{k,1} \circ \cdots \circ v_{k, 2^k-1})_{[0,T)} \; . \]
\end{definition}

For what we need, $T$ is always set to $2n \cdot 2^{2n} = n \cdot 2^{2n+1}$. In other words, the GGM tree will always have height $2n+1$. 

It is known that the output of GGM tree has a small circuit \cite[Lemma 3.2]{CHR23}. For what we need, we note that the assigned value of any vertex in a GGM tree on a given input has a small circuit.

\begin{lemma}\label{lem:GGM_eval}
    Let $\cc{GGMEval}(C, T, x, (i,j))$ denote the $n$-bit assigned value $v_{i,j}$ in the evaluation of the GGM tree $\cc{GGM}_T[C](x)$. There is an algorithm running in $\widetilde{O}(|C|\cdot \log T)$ time that, given $C, T, x, (i,j)$, outputs $\cc{GGMEval}(C, T, x, (i,j))$.
\end{lemma}
\begin{proof}[Proof sketch]
    To compute $v_{i,j}$, it suffices to traverse the GGM tree from the root to the vertex $(i,j)$, applying the circuit $C$ in each step. The running time is clearly bounded by the $\widetilde{O}(|C|\cdot \log T)$ since the GGM tree has height $O(\log T)$.
\end{proof}

\subsection{Modifying Korten's Reduction}
Korten's reduction \cite{Korten} asserts that given a hard truth table $f \notin \Im(\cc{GGM}_T[C])$ and an $\NP$ oracle, one can find a non-image for $C$ in $\poly(T, n)$ time. 

Note that on a fixed $f$, Korten's reduction produces the same output. Hence, if we could efficiently simulate the reduction, we obtain an efficient single-valued algorithm. The remaining parts of this section aim to show that Korten's reduction (after our modification) indeed has a small description.

We modify Korten's reduction in the following manner: instead of traversing the perfect binary tree in a simple bottom-up manner, we perform a \emph{post-order traversal} (i.e. traverse the left subtree, then the right subtree and finally the root). 

For simplicity, we will fix $T$ to be $n\cdot 2^{2n+1}$ and the perfect binary tree has height $2n+1$. We note that this choice of $T$ (i.e. exponential in $n$) is the ``base case'' or ``worst case'' in the iterative win-win argument in \cite{CHR23}. Since we can handle even the ``worst case'', we manage to completely bypass the iterative win-win argument.

\begin{fact}
    In a post-order traversal, any root vertex of a subtree is traversed after all other vertices in the subtree. Any vertex in the right subtree is traversed after all vertices in the left subtree.
\end{fact}


For the ease of presentation, we define the total order $<_P$ for all vertices on a perfect binary tree to be the post-order traversal order. In other words, $u_1 <_P u_2$ if and only if $u_1$ should be traversed before $u_2$.

\begin{fact}
    Given two vertex $u_1 \neq u_2$ in a perfect binary tree, there is an algorithm that decides whether $u_1 <_P u_2$ or $u_2 <_P u_1$ and runs in time linear in the height of the tree. 
\end{fact}
\begin{proof}[Proof sketch]
    The algorithm simply finds the lowest common ancestor $u_a$ of $u_1$ and $u_2$. By definition of the lowest common ancestor, $u_1$ and $u_2$ cannot live in the same proper subtree of $u_a$. 
    
    The vertex living in the left subtree of $u_a$ will be traversed first. If none of them lives in the left subtree of $u_a$, then one of them must be $u_a$ itself. In this case, the vertex living in the right subtree will be traversed first.
\end{proof}

\begin{algorithm}
\caption{$\cc{Korten}'(C,f)$: Modified Korten's reduction}\label{alg:modified_Korten}
\KwIn{$C:\bits^n \rightarrow \bits^{2n}$ denotes the input circuit, and $f \in \bits^T \backslash \Im(\cc{GGM}_T[C])$ denotes the input hard truth table.}
\KwOut{A non-output of $C$.}
\KwData{A perfect binary tree of height $2n+1$ that contains the computational history.}
\For{$j \gets 0$ to $2^{2n+1}-1$}{
    $v_{2n+1,j} \gets f_{[jn, (j+1)n)}$ \tcp*{set $f$ to the leaves}
}
\For{vertex $(i,j)$ in the Post-Order Traversal}{
    Set $v_{i,j}$ be the lexicographically smallest string such that $C(v_{i,j}) = v_{i+1, 2j} \circ v_{i+1, 2j+1}$ 
    \tcp*{this step requires a $\NP$ oracle}
    \If{$v_{i,j}$ does not exist}{
        Set all remaining vertices $\bot$ \;
        \KwRet{$v_{i+1, 2j} \circ v_{i+1, 2j+1}$} \;
    }
}
\KwRet{$\bot$}\;
\end{algorithm}

\subsection{History of \texorpdfstring{$\cc{Korten}'(C,f)$}{Korten'(C,f)}}

The computational history of $\cc{Korten}'(C,f)$ is essentially a \emph{partially-assigned} perfect binary tree of height $2n+1$ where each vertex $(i,j)$ stores a $n$-bit string $v_{i,j}$ or $\bot$. Unlike \cite{CHR23} where they view the computational history as a long string, we shall keep viewing it as a perfect binary tree and exploit the tree structure. Towards that end, we call it $\cc{Histree}(C,f)$. 

\begin{definition}[the computational history of $\cc{Korten}'(C,f)$]
    Let $n, T\in \N$ be such that $T = n \cdot 2^{2n+1}$. Let $C:\bits^n \rightarrow \bits^{2n}$ be a circuit, and $f\in \bits^T$ be a ``hard truth table'' in the sense that $f \notin \Im(\cc{GGM}_T[C])$. The computational history of $\cc{Korten}'(C,f)$, denoted as $\cc{Histree}(C,f)$, is the partially-assigned perfect binary tree obtained by executing $\cc{Korten}'(C,f)$.
\end{definition}

Let $h := \cc{Histree}(C,f)$. For any vertex $u$ in the perfect binary tree, we use $h(u)$ to denote the value $v_u$ stored at $u$. We use $c_L(u)$ and $c_R(u)$ to denote the left child and right child of $u$.

\begin{definition}[Proper left children]
    Given a set of vertices $S$ from a binary tree, we define the set of proper left children of $S$ to be:
        \[ \{u: \exists w\in S, c_L(w) = u, u \notin S \} \; .\]
\end{definition}

The following lemma shows that $h$ has a succinct description.

\begin{lemma}\label{lem:short_description}
    Let $n, T\in \N$ be such that $T = 2n \cdot 2^{2n}$. Let $C:\bits^n \rightarrow \bits^{2n}$ be a circuit, and $f\in \bits^T$. Let $h:=\cc{Histree}(C,f)$. $h$ admits a unique description $D_h$ such that:
    \begin{itemize}
        \item $|D_h| \leq O(n) \cdot \log T$.
        \item There is an algorithm $\cc{Eval}$ that takes in input $D_h$ and any vertex $(i,j)$, outputs $h(i,j)$ in time $\poly(n) \cdot \log T$.
    \end{itemize}
\end{lemma}

\begin{proof}

    We start by describing $D_h$.

    Let $u^* = (i^*, j^*)$ be the vertex where \cref{alg:modified_Korten} finds a solution and terminates. Let $S = \{ u_0 = (0,0), u_1, u_2, \ldots, u^*\}$ be the set of vertices on the unique path starting from the root $(0,0)$ to $u^*$. Note that all vertices in $S$ are assigned $\bot$. 
    
    $D_h$ is defined to contain all \emph{proper left children} of $S$ and the right child of $u^*$, as well as all the stored values in these vertices. We further note that all these vertices have non $\bot$ stored values. In particular, consider any proper left children vertex $u_L$, it lives in the left subtree of its parent while $u^*$ lives in the right subtree. So $u_L$ must have already been traversed when the algorithm terminates. 

    Note that $D_h$ contains both children of $u^*$ and hence the output of $\cc{Korten}'(C,f)$. It is clear from how we construct $D_h$ that $|D_h| \leq O(n) \cdot \log T$ since any path on the tree contains $O(\log T)$ vertices and every vertex stored in $D_h$ carries $O(n)$ bits of information. Also, $D_h$ is uniquely defined for any fixed $h$.

    Next, we show how to efficiently evaluate $h(i,j)$ given any vertex $(i,j)$ in the perfect binary tree. Notice that any subtree in $h$ is also a (smaller) GGM tree. From \cref{lem:GGM_eval}, if we know the stored value of any ancestor of $(i,j)$, we can efficiently (in time $\poly(n) \cdot \log T$) evaluate $h(i,j)$.
    
    Therefore, it suffices to show that for any $(i,j)$, one of its (non $\bot$) ancestors is stored in $D_h$:
    \begin{enumerate}
        \item If $u^*$ is an ancestor of $(i,j)$, then one of $u^*$'s children is an ancestor of $(i,j)$ and we know both children are stored in $D_h$.
        \item If $(i,j)$ is an ancestor of $u^*$, then $h(i,j) = \bot$.
        \item Otherwise, let $u_a$ be the lowest common ancestor of $u^*$ and $(i,j)$, and we know that $u_a \in S$. If $(i,j)$ falls in the left subtree of $u_a$, then the left child of $u_a$ is an ancestor of $(i,j)$ which is stored in $D_h$. If $(i,j)$ falls in the right subtree of $u_a$ then we argue that $h(i,j) = \bot$. This is because the algorithm stopped at $u^*$ living in the left subtree of $u_a$, and would not have traversed $(i,j)$. 
    \end{enumerate}
    This concludes the proof.
\end{proof}

The final ingredient we need is that $D_h$ admits a $\Pi_1$ verifier on whether its corresponding computational history $h$ is the correct one.
\begin{lemma}[$\Pi_1$ verification of the history]\label{lem:verify_history}
    Let $h := \cc{Histree}(C,f)$.
        There is an oracle algorithm $V$ with input parameter $T, n$ such that the following holds:
        \begin{enumerate}
            \item $V$ takes $\tilde{f} \in \bits^T$ as an oracle, $C$, $\widetilde{D_h}$ and $w\in \bits^{2 \log T + n}$ as inputs. It runs in $\poly(n)$ time.
            \item $D_h$ defined on $h := \cc{Histree}(C,f)$ is the unique string satisfying the following:
            \[ V^{f}(C, D_h, w) = 1, \hspace{2em} \forall w\ \in \bits^{2\log T + n} \; . \]
    \end{enumerate}
\end{lemma}
\begin{proof}
    The verifier proceeds in two parts. First part consists of parsing and checking whether $\widetilde{D_h}$ is indeed generated from a valid post-order traversal on the perfect binary tree. 
    
    In particular, it reads the terminating vertex $u^*$ based on the two children of $u^*$ stored in $\widetilde{D_h}$, finds the path from $(0,0)$ to $u^*$ and check that all proper left children of vertices on the path are included in $\widetilde{D_h}$ (and of course the right child of $u^*$ should be included). Also stored values of these vertices are non $\bot$.
    

    Upon passing the first part of the verification, we know $\widetilde{D_h}$ corresponds to some \emph{partially-assigned} perfect binary tree $\tilde{h}$ and it remains to check that $\tilde{h}$ is the computational history $\cc{Histree}(C,f)$. One should think of the verifier making at most $2^{|w|}$ checks and accepts only if all $2^{|w|}$ check passes. 
    
    The verifier $V$ needs to make the following checks. Note that whenever we need some value $v_{i,j}$, we will call $\cc{Eval}(\widetilde{D_h}, (i,j))$ for the value. Recall that the total order $<_P$ is defined according to the post-order traversal sequence. 
        \begin{enumerate}
            \item The values written on the leaves are indeed $f$. Hence, for every $j \in [0, 2^{2n+1}-1]$, check that $v_{2n+1, j}$ is consistent with the corresponding string in $f$.
            \item For every $(i,j) <_{P} u^*$, $C(v_{i,j}) = v_{i+1, 2j} \circ v_{i+1, 2j+1}$. (the values are consistent with the children)
            \item For every $(i,j) <_{P} u^*$, for every $x \in \bits^n$ that is lexicographically smaller than $v_{i,j}$, $C(x) \neq v_{i+1, 2j} \circ v_{i+1, 2j+1}$. (the lexicographically first requirement)
            \item Let $(i^*, j^*) = u^*$, then for every $x \in \bits^n$, $C(x) \neq v_{i^*+1, 2j^*} \circ v_{i^*+1, 2j^*+1}$. (the two children of $u^*$ form a non-image of $C$)
            \item For every $(i,j)$ where $u^* \le_{P} (i,j)$, $v_{i,j} = \bot$.
        \end{enumerate}
    Each of the above checks is local (requires assigned values of at most $3$ vertices) and efficient (runs in time $\poly(n, \log T)$). There are in total $O(T)$ vertices and therefore $O(T)$ tests, which can be implemented with a universal ($\forall$) quantification over at most $2 \log T + n$ bits. 
    
    Clearly the correct history $h$ (and therefore its unique description $D_h$) passes all these checks. Also these checks uniquely determine $h$ as they are essentially enforcing every step of execution of $\cc{Korten}'(C,f)$. 

\end{proof}

\section{Circuit Lower Bound for \texorpdfstring{$\cc{S_2E}$}{S2E}}
\subsection{Single-valued \texorpdfstring{$\fsigp$}{FSigma2P} Algorithm}

We start by showing a simple single-valued $\fsigp$ Algorithm for $\cc{Avoid}$.
\begin{theorem}
    There is a single-valued $\fsigp$ algorithm $A$: when given any circuit $C:\bits^n \rightarrow \bits^{2n}$ as input, $A(C)$ outputs $y_C$ such that $y_C \notin \Im(C)$ .
\end{theorem}
\begin{proof}
    On input a circuit $C:\bits^n \rightarrow \bits^{2n}$, let $T = 2n \cdot 2^{2n}$ and $f \in \bits^{T}$ be the concatenation of all $2n$-length bit strings. Let $h = \cc{Histree}(C,f)$. $V_A(C, \pi_1, \pi_2)$ is defined as follows: it parses $\pi_1$ as $D_h$ and $\pi_2$ as $w$, simulates the verifier $V^f(C, D_h, w)$ in \cref{lem:verify_history}. It outputs the non-image of $C$ stored in $D_h$ iff $V^f(C, D_h, w) = 1$. Otherwise it outputs $\bot$.
    
    
    Note that every position of $f$ can be easily computed since it is just enumerating all $2n$-length strings. Hence the simulation can be done in polynomial time.
\end{proof}

\subsection{Single-valued \texorpdfstring{$\fsp$}{FS2P} Algorithm}

In order to generalise the $\fsigp$ algorithm above to a $\fsp$ algorithm, we need a `selector' that chooses the correct $D_h$ when two candidates are given. We formalise such selector in the following lemma.

\begin{lemma}\label{lem:selector}
    Let $n, T\in \N$ be such that $T = 2n \cdot 2^{2n}$. Let $C:\bits^n \rightarrow \bits^{2n}$ be a circuit, and $f\in \bits^T$. Let $h:=\cc{Histree}(C,f)$ and $D_h$ be the succinct description of $h$ defined in \cref{lem:short_description}. Given $f$ as an oracle and two strings $\pi_1, \pi_2$ as additional input, with the promise that $\pi_i = D_h$ for at least one $i \in \{1,2\}$, there is a deterministic algorithm $S$ such that $S^f(C, \pi_1, \pi_2) = \pi_i$ and runs in time $\poly(n) \cdot \log T$.
\end{lemma}
\begin{proof}
    $S$ starts by parsing $\pi_1, \pi_2$ as $D_h$. If any of them fail to parse (i.e. the vertices are not derived from a post-order traversal), $S$ simply discards it and output the other one.
    
    Let $h_1$ and $h_2$ be the corresponding perfect binary tree generated from $\pi_1$ and $\pi_2$. Let $(i^*_1, j^*_1)$ be the termination vertex in $h_1$ and $(i^*_2, j^*_2)$ be the termination vertex in $h_2$. $S$ will efficiently find a single vertex that contains different stored values in $h_1$ and $h_2$. In particular, we consider two cases:
    \begin{enumerate}
        \item $(i^*_1, j^*_1) \neq (i^*_2, j^*_2)$: Without loss of generality, we may assume $(i^*_1, j^*_1) <_P (i^*_2, j^*_2)$. Then we know that $h_2(i^*_1, j^*_1)\neq \bot$. 
        \item $(i^*_1, j^*_1) = (i^*_2, j^*_2)$: Then they store the same set of vertices in $\pi_1$ and $\pi_2$, and one of the vertex must have a different stored value. 
    \end{enumerate}

    Now given a vertex (say $u$) where $h_1(u) \neq h_2(u)$, $S$ proceeds as follows:
    \begin{enumerate}
        \item If $u$ is a leaf vertex, then $S$ checks it against $f$ and decides which is the correct $D_h$.
        \item Otherwise, $S$ checks if $C(h_1(u)) = C(h_2(u))$. If they are indeed pre-images of the same value, $S$ picks $\pi_i$ for $i\in \{1,2\}$ such that $h_i(u)$ is the lexicographically smaller one. 
        
        If $C(h_1(u)) \neq C(h_2(u))$ or if $h_2(u) \neq h_1(u) = \bot$, then at least one of $u$'s children (say $u'$) should have a different stored value. $S$ then repeats the whole procedure on $u'$.
    \end{enumerate}
    
    It is clear from the description that $S$ terminates in $O(\log T)$ recursive steps, and the overall running time is $\poly(n) \cdot \log T$.
\end{proof}

\begin{theorem}\label{thm:n_2n_fsp}
    There is a single-valued $\fsp$ algorithm $A$: when given any circuit $C:\bits^n \rightarrow \bits^{2n}$ as input, $A(C)$ outputs $y_C$ such that $y_C \notin \Im(C)$.
\end{theorem}
\begin{proof}
    On input a circuit $C:\bits^n \rightarrow \bits^{2n}$, let $T = 2n \cdot 2^{2n}$ and $f \in \bits^{T}$ be the concatenation of all $2n$-length bit strings. Let $h = \cc{Histree}(C,f)$. $V_A(C, \pi_1, \pi_2)$ is defined as follows: it applies the selector algorithm $S^f(C, \pi_1, \pi_2)$ from \cref{lem:selector} and obtains $\pi_i = D_h$. 
    It then outputs the non-image of $C$ stored in $\pi_i$.
    
    
    Note that every position of $f$ can be easily computed since it is just enumerating all $2n$-length strings. Hence the simulation can be done in polynomial time.
\end{proof}

\subsection{Circuit Lower Bound}
Before we get to our circuit lower bound, we need a few results from \cite{Korten,CHR23}:

\begin{theorem}\cite[Theorem 2.3]{CHR23}\label{thm:post_process}
    Let $A(x)$ be a single-valued $\fsp$ algorithm and $B(x,y)$ be an $\cc{FP}^{\NP}$ algorithm, both with fixed output length. The function $f(x):=B(x,A(x))$ also admits an $\fsp$ algorithm.
\end{theorem}

\begin{lemma}\cite[Lemma 3]{Korten}\label{lem:stretch_increase}
    Let $n\in \N$. There is a polynomial time algorithm $A$ and an $\cc{FP}^{\NP}$ algorithm $B$ such that the following holds:
    \begin{enumerate}
        \item Given a circuit $C:\bits^n \rightarrow \bits^{n+1}$, $A(C)$ outputs a circuit $D:\bits^n \rightarrow \bits^{2n}$.
        \item Given any $y\in \bits^{2n} \backslash \Im(D)$, $B(C,y)$ outputs a string $z \in \bits^{n+1} \backslash \Im(C)$.
    \end{enumerate}
\end{lemma}

\begin{definition}\cite[Section 2.3]{CHR23}
     For $n, s\in \N$ where $n \leq s \leq 2^n$, the truth table generator circuit $\cc{TT}_{n,s}:\bits^{L_{n,s}} \rightarrow \bits^{2^n}$ maps a stack program of description size $L_{n,s} = (s+1)(7+\log(n+s))$ into its truth table. Moreover, such circuit can be uniformly constructed in time $\poly(2^n)$.
\end{definition}


The following corollary follows from \cref{thm:n_2n_fsp,thm:post_process} and 
\cref{lem:stretch_increase}.

\begin{corollary}\label{cor:main}
    There is a single-valued $\fsp$ algorithm $A$: when given any circuit $C:\bits^n \rightarrow \bits^{n+1}$ as input, $A(C)$ outputs $y_C$ such that $y_C \notin \Im(C)$ .
\end{corollary}

\begin{corollary}
    $\cc{S_2E} \not\subset i.o.$-$\cc{SIZE}[2^n/n]$.
\end{corollary}
\begin{proof}[Proof Sketch]
    Let $A$ be the single-valued algorithm from \cref{cor:main} and set $s := 2^n / n$. Define the language $\mathcal{L}$ such that the truth table of the characteristic function of $\mathcal{L} \cap \bits^n$ is $A(\cc{TT}_{n,s})$. By our choice of $s$, $L_{n,s} = (s+1)(7+\log(n+s)) < 2^n$ and hence $\cc{TT}_{n,s}$ is a valid $\cc{Avoid}$ instance.
    
    $\mathcal{L} \notin i.o.$-$\cc{SIZE}[s(n)]$ since any $s$-size $n$-input circuit $C$ can be encoded into a stack program of size $L_{n,s}$ bits \cite{FM05}.
    
    $\mathcal{L} \in \cc{S_2E}$ since one can compute the truth table using algorithm $A$.
\end{proof}


\begin{remark}
    Similar to \cite{CHR23}, it is not hard to verify that all our results above relativise.
\end{remark}

\section*{Acknowledgements}
The author was supported by the NUS-NCS Joint Laboratory for Cyber Security, Singapore.

The author would like to thank Eldon Chung, Karthik Gajulapalli, Alexander Golovnev, Sidhant Saraogi and Noah Stephens-Davidowitz for inspiring discussions. 

The author would also like to thank Divesh Aggarwal, Lijie Chen, Alexander Golovnev, Hanlin Ren and Sidhant Saraogi for their comments on early drafts of this work. 

The author would also like to thank anonymous STOC reviewers for helpful comments.

\bibliographystyle{alpha}
\bibliography{biblio}

\begin{thebibliography}{KKMP21}

\bibitem[AB09]{AB09}
Sanjeev Arora and Boaz Barak.
\newblock {\em Computational Complexity: A Modern Approach}.
\newblock Cambridge University Press, USA, 1st edition, 2009.

\bibitem[Cai07]{Cai07}
Jin-Yi Cai.
\newblock $\cc{S}^p_2 \subseteq \cc{ZPP}^{\NP}$.
\newblock {\em Journal of Computer and System Sciences}, 73, 02 2007.

\bibitem[CHLR23]{CHLR23}
Yeyuan Chen, Yizhi Huang, Jiatu Li, and Hanlin Ren.
\newblock Range avoidance, remote point, and hard partial truth table via satisfying-pairs algorithms.
\newblock In {\em Proceedings of the 55th Annual ACM Symposium on Theory of Computing}, STOC 2023, page 1058–1066, New York, NY, USA, 2023. Association for Computing Machinery.

\bibitem[CHR24]{CHR23}
Lijie Chen, Shuichi Hirahara, and Hanlin Ren.
\newblock Symmetric exponential time requires near-maximum circuit size.
\newblock In {\em Proceedings of the 56th Annual ACM SIGACT Symposium on Theory of Computing}, STOC 2024, to appear. Association for Computing Machinery, 2024.

\bibitem[CT22]{CT21}
Lijie Chen and Roei Tell.
\newblock Hardness vs randomness, revised: Uniform, non-black-box, and instance-wise.
\newblock In {\em 2021 IEEE 62nd Annual Symposium on Foundations of Computer Science (FOCS)}, pages 125--136, 2022.

\bibitem[CZ19]{CZ19}
Eshan Chattopadhyay and David Zuckerman.
\newblock {Explicit two-source extractors and resilient functions}.
\newblock {\em Annals of Mathematics}, 189(3):653 -- 705, 2019.

\bibitem[FM05]{FM05}
Gudmund~Skovbjerg Frandsen and Peter~Bro Miltersen.
\newblock Reviewing bounds on the circuit size of the hardest functions.
\newblock {\em Information Processing Letters}, 95(2):354--357, 2005.

\bibitem[GG11]{GG11}
Erann Gat and Shafi Goldwasser.
\newblock Probabilistic search algorithms with unique answers and their cryptographic applications.
\newblock {\em Electron. Colloquium Comput. Complex.}, TR11, 2011.

\bibitem[GGM86]{GGM86}
Oded Goldreich, Shafi Goldwasser, and Silvio Micali.
\newblock How to construct random functions.
\newblock {\em J. ACM}, 33(4):792–807, aug 1986.

\bibitem[GGNS23]{GGNS23}
Karthik Gajulapalli, Alexander Golovnev, Satyajeet Nagargoje, and Sidhant Saraogi.
\newblock Range avoidance for constant depth circuits: Hardness and algorithms.
\newblock In Nicole Megow and Adam~D. Smith, editors, {\em Approximation, Randomization, and Combinatorial Optimization. Algorithms and Techniques, {APPROX/RANDOM} 2023, September 11-13, 2023, Atlanta, Georgia, {USA}}, volume 275 of {\em LIPIcs}, pages 65:1--65:18. Schloss Dagstuhl - Leibniz-Zentrum f{\"{u}}r Informatik, 2023.

\bibitem[GLW22]{GLW22}
Venkatesan Guruswami, Xin Lyu, and Xiuhan Wang.
\newblock Range avoidance for low-depth circuits and connections to pseudorandomness.
\newblock In Amit Chakrabarti and Chaitanya Swamy, editors, {\em Approximation, Randomization, and Combinatorial Optimization. Algorithms and Techniques, {APPROX/RANDOM} 2022, September 19-21, 2022, University of Illinois, Urbana-Champaign, {USA} (Virtual Conference)}, volume 245 of {\em LIPIcs}, pages 20:1--20:21. Schloss Dagstuhl - Leibniz-Zentrum f{\"{u}}r Informatik, 2022.

\bibitem[Gol08]{Gol08}
Oded Goldreich.
\newblock Computational complexity: A conceptual perspective.
\newblock {\em SIGACT News}, 39(3):35–39, sep 2008.

\bibitem[HNOS96]{HNOS96}
Lane~A. Hemaspaandra, Ashish~V. Naik, Mitsunori Ogihara, and Alan~L. Selman.
\newblock Computing solutions uniquely collapses the polynomial hierarchy.
\newblock {\em SIAM Journal on Computing}, 25(4):697--708, 1996.

\bibitem[ILW23]{ILW23}
Rahul Ilango, Jiatu Li, and R.~Ryan Williams.
\newblock Indistinguishability obfuscation, range avoidance, and bounded arithmetic.
\newblock In {\em Proceedings of the 55th Annual ACM Symposium on Theory of Computing}, STOC 2023, page 1076–1089, New York, NY, USA, 2023. Association for Computing Machinery.

\bibitem[IW97]{IW97}
Russell Impagliazzo and Avi Wigderson.
\newblock P = bpp if e requires exponential circuits: Derandomizing the xor lemma.
\newblock In {\em Proceedings of the Twenty-Ninth Annual ACM Symposium on Theory of Computing}, STOC '97, page 220–229, New York, NY, USA, 1997. Association for Computing Machinery.

\bibitem[Kan82]{Kannan82}
Ravindran Kannan.
\newblock Circuit-size lower bounds and non-reducibility to sparse sets.
\newblock {\em Information and Control}, 55(1):40--56, 1982.

\bibitem[KKMP21]{KKMP21}
Robert Kleinberg, Oliver Korten, Daniel Mitropolsky, and Christos Papadimitriou.
\newblock {Total Functions in the Polynomial Hierarchy}.
\newblock In James~R. Lee, editor, {\em 12th Innovations in Theoretical Computer Science Conference (ITCS 2021)}, volume 185 of {\em Leibniz International Proceedings in Informatics (LIPIcs)}, pages 44:1--44:18, Dagstuhl, Germany, 2021. Schloss Dagstuhl--Leibniz-Zentrum f{\"u}r Informatik.

\bibitem[Kor21]{Korten}
Oliver Korten.
\newblock The hardest explicit construction.
\newblock In {\em 2021 IEEE 62nd Annual Symposium on Foundations of Computer Science (FOCS)}, pages 433--444, 2021.

\bibitem[Li23]{Li23}
Xin Li.
\newblock Two source extractors for asymptotically optimal entropy, and (many) more.
\newblock In {\em 2023 IEEE 64th Annual Symposium on Foundations of Computer Science (FOCS)}, pages 1271--1281, Los Alamitos, CA, USA, nov 2023. IEEE Computer Society.

\bibitem[MVW99]{MVW99}
Peter~Bro Miltersen, N.~V. Vinodchandran, and Osamu Watanabe.
\newblock Super-polynomial versus half-exponential circuit size in the exponential hierarchy.
\newblock In Takano Asano, Hideki Imai, D.~T. Lee, Shin-ichi Nakano, and Takeshi Tokuyama, editors, {\em Computing and Combinatorics}, pages 210--220, Berlin, Heidelberg, 1999. Springer Berlin Heidelberg.

\bibitem[NW94]{NW94}
Noam Nisan and Avi Wigderson.
\newblock Hardness vs randomness.
\newblock {\em Journal of Computer and System Sciences}, 49(2):149--167, 1994.

\bibitem[Rad21]{Radziszowski2021}
Stanisław~P. Radziszowski.
\newblock Small ramsey numbers.
\newblock {\em The Electronic Journal of Combinatorics [electronic only]}, DS01, 2021.

\bibitem[RSW22]{RSW22}
Hanlin Ren, Rahul Santhanam, and Zhikun Wang.
\newblock On the range avoidance problem for circuits.
\newblock In {\em 2022 IEEE 63rd Annual Symposium on Foundations of Computer Science (FOCS)}, pages 640--650, Los Alamitos, CA, USA, nov 2022. IEEE Computer Society.

\bibitem[Sha49]{Shannon}
Claude.~E. Shannon.
\newblock The synthesis of two-terminal switching circuits.
\newblock {\em The Bell System Technical Journal}, 28(1):59--98, 1949.

\bibitem[VW23]{VW23}
Nikhil Vyas and Ryan Williams.
\newblock {On Oracles and Algorithmic Methods for Proving Lower Bounds}.
\newblock In Yael Tauman~Kalai, editor, {\em 14th Innovations in Theoretical Computer Science Conference (ITCS 2023)}, volume 251 of {\em Leibniz International Proceedings in Informatics (LIPIcs)}, pages 99:1--99:26, Dagstuhl, Germany, 2023. Schloss Dagstuhl -- Leibniz-Zentrum f{\"u}r Informatik.

\end{thebibliography}

\end{document}